\documentclass[11pt]{article}%
\usepackage{amsmath}
\usepackage{amsfonts}
\usepackage{amssymb}
\usepackage{graphicx}
\usepackage[ colorlinks = true,              linkcolor = blue, urlcolor  =
blue,              citecolor = red,              anchorcolor = green,
]{hyperref}
\usepackage{mathtools}%
\providecommand{\U}[1]{\protect\rule{.1in}{.1in}}
\newtheorem{theorem}{Theorem}

\newtheorem{corollary}[theorem]{Corollary}

\newtheorem{lemma}[theorem]{Lemma}

\newtheorem{proposition}[theorem]{Proposition}

\newenvironment{proof}[1][Proof]{\noindent\textbf{#1.} }{\ \rule{0.5em}{0.5em}}
\allowdisplaybreaks
\begin{document}

\title{\textbf{On distinguishability distillation and dilution exponents}}
\author{Mark\ M.~Wilde\thanks{Hearne Institute for Theoretical Physics, Department of
Physics and Astronomy, and Center for Computation and Technology, Louisiana
State University, Baton Rouge, Louisiana 70803, USA, Email: mwilde@lsu.edu}}
\date{\today}
\maketitle

\begin{abstract}
In this note, I define error exponents and strong converse exponents for the
tasks of distinguishability distillation and dilution. These are counterparts
to the one-shot distillable distinguishability and the one-shot
distinguishability cost, as previously defined in the resource theory of
asymmetric distinguishability. I show that they can be evaluated by
semi-definite programming, establish a number of their properties, bound them
using R\'{e}nyi relative entropies, and relate them to each other.

\end{abstract}
\tableofcontents

\section{Introduction to basic operational quantities}

The resource theory of asymmetric distinguishability was proposed and
developed in \cite{Matsumoto10,Matsumoto11,Wang2019states,Wang2019channels}.
The basic operational tasks are known as distinguishability distillation and
distinguishability dilution, in which the goal of distillation is to convert a
pair of states to a pure form of distinguishability known as bits of
asymmetric distinguishability (AD), and the goal of dilution is to accomplish
the reverse task. In both cases, we are interested in these processes
occurring as efficiently as possible.

The focus of the one-shot operational quantities proposed in
\cite{Wang2019states} is to fix the transformation error and then, for distillation, to
maximize the number of bits of AD that can be extracted from a pair of states,
and, for dilution, to minimize the number of bits of AD needed to generate a
pair of states. Here, I flip the objective of the task around, and instead
place a threshold on the number of bits of AD allowed and then minimize the
transformation error that can be realized with this constraint in place. The
resulting quantities are known as error exponents and strong converse
exponents, similar to what has been studied for a long time in hypothesis
testing and information theory \cite{1055254}.

\subsection{Distinguishability distillation}

The distillable distinguishability of the states $\rho$ and $\sigma$ is
defined as \cite{Wang2019states}
\begin{equation}
D_{d}^{\varepsilon}(\rho,\sigma)\coloneqq \log_{2}\sup_{\mathcal{P}%
\in\text{CPTP}}\{M:\mathcal{P}(\rho)\approx_{\varepsilon}|0\rangle
\!\langle0|,\mathcal{P}(\sigma)=\pi_{M}\},
\end{equation}
where CPTP stands for the set of completely positive, trace-preserving maps
(quantum channels), $M\geq1$,%
\begin{equation}
\pi_{M}\coloneqq \frac{1}{M}|0\rangle\!\langle0|+\left(  1-\frac{1}{M}\right)
|1\rangle\!\langle1|,
\end{equation}
and the shorthand $\approx_{\varepsilon}$ means the following:%
\begin{equation}
\tau\approx_{\varepsilon}\omega\qquad\Longleftrightarrow\qquad\frac{1}%
{2}\left\Vert \tau-\omega\right\Vert _{1}\leq\varepsilon.
\end{equation}
It is known that \cite{Wang2019states}
\begin{equation}
D_{d}^{\varepsilon}(\rho,\sigma)=D_{\min}^{\varepsilon}(\rho\Vert\sigma),
\end{equation}
where the smooth min-relative entropy is defined as \cite{BD10,BD11,WR12}
\begin{equation}
D_{\min}^{\varepsilon}(\rho\Vert\sigma)\coloneqq -\log_{2}\inf_{\Lambda\geq
0}\left\{  \operatorname{Tr}[\Lambda\sigma]:\operatorname{Tr}[\Lambda\rho
]\geq1-\varepsilon,\ \Lambda\leq I\right\}  . \label{eq:primal-SDP-smooth-min}%
\end{equation}
This quantity is also known as hypothesis testing relative entropy \cite{WR12}
and can be computed by semi-definite programming \cite{DKFRR12}. By strong
duality, we have the following alternate expression for $D_{\min}%
^{\varepsilon}(\rho\Vert\sigma)$ \cite{DKFRR12,Wang2019states}:%
\begin{equation}
D_{\min}^{\varepsilon}(\rho\Vert\sigma)=-\log_{2}\sup_{\mu,X\geq0}\left\{
\mu\left(  1-\varepsilon\right)  -\operatorname{Tr}[X]:\mu\rho\leq
\sigma+X\right\}  . \label{eq:dual-SDP-smooth-min}%
\end{equation}

We can turn the objectives of this task around and define the following
one-shot quantities for $m\geq0$:%
\begin{align}
E_{d}^{m}(\rho\Vert\sigma)  &  \coloneqq-\log_{2}\inf_{\mathcal{P}%
\in\text{CPTP}}\left\{  \varepsilon:\mathcal{P}(\rho)\approx_{\varepsilon
}|0\rangle\!\langle0|,\ \mathcal{P}(\sigma)=\pi_{2^{m^{\prime}}},\ m^{\prime
}\geq m\right\}  ,\label{eq:err-exp-div-def}\\
\widetilde{E}_{d}^{m}(\rho\Vert\sigma)  &  \coloneqq-\log_{2}\sup
_{\mathcal{P}\in\text{CPTP}}\left\{  1-\varepsilon:\mathcal{P}(\rho
)\approx_{\varepsilon}|0\rangle\!\langle0|,\ \mathcal{P}(\sigma)=\pi
_{2^{m^{\prime}}},\ m^{\prime}\geq m\right\}  . \label{eq:sc-exp-div-def}%
\end{align}
By definition, it follows that%
\begin{equation}
2^{-E_{d}^{m}(\rho\Vert\sigma)}=1-2^{-\widetilde{E}_{d}^{m}(\rho\Vert\sigma)}.
\label{eq:dist-dist-err-exp-sc-exp-eq}%
\end{equation}
Proposition~\ref{prop:sdp-err-exp-div} and Corollary~\ref{prop:sc-reduction}
below give simpler expressions for these quantities that can be evaluated by
semi-definite programming. These simplified expressions are then helpful for
establishing bounds on these quantities in terms of R\'{e}nyi relative
entropies (see Propositions~\ref{prop:dist-sc-exp-to-renyi} and
\ref{prop:dist-err-exp-to-renyi}).

The idea here is that we are trying to distill at least $m$ bits of asymmetric
distinguishability, in the sense of \cite{Wang2019states}, and we would like
to minimize the transformation error subject to this constraint. Let us call the first
quantity in \eqref{eq:err-exp-div-def} the distillation error exponent and the
second quantity in \eqref{eq:sc-exp-div-def} the distillation strong converse exponent.

For the i.i.d.~case, there is an \textquotedblleft asymptotic equipartition
property\textquotedblright\ or \textquotedblleft large deviation
property\textquotedblright\ as follows:%
\begin{align}
\lim_{n\rightarrow\infty}\frac{1}{n}E_{d}^{nR}(\rho^{\otimes n},\sigma
^{\otimes n})  &  =\sup_{\alpha\in\left(  0,1\right)  }\left(  \frac{\alpha
-1}{\alpha}\right)  \left(  R-D_{\alpha}(\rho\Vert\sigma)\right)
,\label{eq:err-exp-hyp-test}\\
\lim_{n\rightarrow\infty}\frac{1}{n}\widetilde{E}_{d}^{nR}(\rho^{\otimes
n},\sigma^{\otimes n})  &  =\sup_{\alpha\in\left(  1,\infty\right)  }\left(
\frac{\alpha-1}{\alpha}\right)  \left(  R-\widetilde{D}_{\alpha}(\rho
\Vert\sigma)\right)  , \label{eq:str-conv-exp-hyp-test}%
\end{align}
where the respective Petz- \cite{P85,P86} and sandwiched
\cite{muller2013quantum,WWY14} R\'enyi relative entropies are defined as
\begin{align}
D_{\alpha}(\rho\Vert\sigma)  &  \coloneqq \frac{1}{\alpha-1}\log
_{2}\operatorname{Tr}[\rho^{\alpha}\sigma^{1-\alpha}],\\
\widetilde{D}_{\alpha}(\rho\Vert\sigma)  &  \coloneqq \frac{1}{\alpha-1}%
\log_{2}\operatorname{Tr}[(\sigma^{\left(  1-\alpha\right)  /2\alpha}%
\rho\sigma^{\left(  1-\alpha\right)  /2\alpha})^{\alpha}].
\end{align}
The equalities in \eqref{eq:err-exp-hyp-test} and
\eqref{eq:str-conv-exp-hyp-test} follow from \cite{N06,Hay07,MO13}, as well as
the simple reductions in Proposition~\ref{prop:sdp-err-exp-div} and
Corollary~\ref{prop:sc-reduction} below. The first asymptotic quantity is only
meaningful when $R<D(\rho\Vert\sigma)$ and the second only when $R>D(\rho
\Vert\sigma)$, where the quantum relative entropy is defined as \cite{U62}
\begin{equation}
D(\rho\Vert\sigma)\coloneqq \operatorname{Tr}[\rho(\log_{2}\rho-\log_{2}%
\sigma)].
\end{equation}

\subsection{Distinguishability cost}

The one-shot distinguishability cost is defined as \cite{Wang2019states}
\begin{equation}
D_{c}^{\varepsilon}(\rho,\sigma)\coloneqq \log_{2}\inf_{\mathcal{P}%
\in\text{CPTP}}\{M:\mathcal{P}(|0\rangle\!\langle0|)\approx_{\varepsilon}%
\rho,\mathcal{P}(\pi_{M})=\sigma\},
\end{equation}
and it is known that \cite{Wang2019states}
\begin{equation}
D_{c}^{\varepsilon}(\rho,\sigma)=D_{\max}^{\varepsilon}(\rho\Vert\sigma),
\label{eq:smooth-max-op-int}%
\end{equation}
where the smooth max-relative entropy \cite{D09} is defined as%
\begin{align}
D_{\max}^{\varepsilon}(\rho\Vert\sigma)  &  \coloneqq \log_{2}\inf
_{\lambda,\widetilde{\rho}\geq0}\left\{  \lambda:\widetilde{\rho}\leq
\lambda\sigma,\ \frac{1}{2}\left\Vert \rho-\widetilde{\rho}\right\Vert
_{1}\leq\varepsilon,\ \operatorname{Tr}[\widetilde{\rho}]=1\right\} \\
&  =\log_{2}\inf_{\lambda,\widetilde{\rho},Z\geq0}\left\{  \lambda
:\widetilde{\rho}\leq\lambda\sigma,\ \rho-\widetilde{\rho}\leq
Z,\ \operatorname{Tr}[Z]\leq\varepsilon,\ \operatorname{Tr}[\widetilde{\rho
}]=1\right\}  . \label{eq:smooth-dmax-SDP}%
\end{align}
The equality in \eqref{eq:smooth-max-op-int} endows the smooth max-relative
entropy with its fundamental operational meaning as one-shot
distinguishability cost \cite{Wang2019states}. Eq.~\eqref{eq:smooth-dmax-SDP}
clarifies that the smooth max-relative entropy can be computed by
semi-definite programming \cite{Wang2019states}. By strong duality, the
following equality holds \cite{Wang2019states}
\begin{multline}
D_{\max}^{\varepsilon}(\rho\Vert\sigma)=\\
\log_{2}\sup_{t,X,Q\geq0,\mu\in\mathbb{R}}\left\{  \operatorname{Tr}%
[Q\rho]+\mu-\varepsilon t:\operatorname{Tr}[X\sigma]\leq1,\ Q\leq tI,\ Q+\mu
I\leq X\right\}  .
\end{multline}
It follows that the constraint $\operatorname{Tr}[X\sigma]\leq1$ can be
saturated with equality because the objective function only increases under
this change.

For the cost problem, we can turn the objective around, as we did for the
distillation problem, to arrive at the following quantities, i.e., dilution
error exponent and dilution strong converse exponent, respectively:%
\begin{align}
E_{c}^{m}(\rho\Vert\sigma)  &  \coloneqq-\log_{2}\inf_{\mathcal{P}%
\in\text{CPTP}}\left\{  \varepsilon:\mathcal{P}(|0\rangle\!\langle
0|)\approx_{\varepsilon}\rho,\mathcal{P}(\pi_{2^{m^{\prime}}})=\sigma
,\ m^{\prime}\leq m\right\}  ,\\
\widetilde{E}_{c}^{m}(\rho\Vert\sigma)  &  \coloneqq-\log_{2}\sup
_{\mathcal{P}\in\text{CPTP}}\left\{  1-\varepsilon:\mathcal{P}(|0\rangle
\!\langle0|)\approx_{\varepsilon}\rho,\mathcal{P}(\pi_{2^{m^{\prime}}}%
)=\sigma,\ m^{\prime}\leq m\right\}  .
\end{align}
The idea is that we are trying to use no more than $m$ bits of asymmetric
distinguishability to generate the pair $\left(  \rho,\sigma\right)  $ with as
small an error as possible. By definition, it follows that%
\begin{equation}
2^{-E_{c}^{m}(\rho\Vert\sigma)}=1-2^{-\widetilde{E}_{c}^{m}(\rho\Vert\sigma)}.
\label{eq:relate-err-exp-sc-exp-cost}%
\end{equation}
Proposition~\ref{prop:dist-dil-sdp} and Corollary~\ref{cor:sc-exp-dist-dil}
give simpler expressions for these quantities that can be evaluated by
semi-definite programming. Later on, I establish bounds on these exponents in
terms of R\'{e}nyi relative entropies (see
Propositions~\ref{prop:one-shot-lower-bnd-dist-dil} and
\ref{prop:err-exp-cost-to-renyi}).

It is then of interest to determine the following asymptotic operational
quantities:
\begin{align}
&  \lim_{n\rightarrow\infty}\frac{1}{n}E_{c}^{nR}(\rho^{\otimes n}%
,\sigma^{\otimes n}),\\
&  \lim_{n\rightarrow\infty}\frac{1}{n}\widetilde{E}_{c}^{nR}(\rho^{\otimes
n},\sigma^{\otimes n}).
\end{align}
See Section~\ref{sec:rec-devs} for a discussion of recent developments in this regard.

Here I also develop basic properties of and establish relationships between
the quantities $E_{d}^{m}(\rho\Vert\sigma)$, $\widetilde{E}_{d}^{m}(\rho
\Vert\sigma)$, $E_{c}^{m}(\rho\Vert\sigma)$, and $\widetilde{E}_{c}^{m}%
(\rho\Vert\sigma)$, to treat them as relative entropies in their own right.
For example, each of them obeys the data-processing inequality, which holds by
means of an operational argument. This is analogous to how the smooth min- and
max-relative entropies have been traditionally regarded as relative entropies,
even though their true origin is in the operational tasks of distinguishability distillation and dilution, respectively.

\section{Distinguishability distillation}

\subsection{Semi-definite programs}

A brief review of semi-definite programming (SDP) is available in
Appendix~\ref{sec:background-SDPs}.

The first claim consists of the following SDP expressions for $E_{d}^{m}%
(\rho\Vert\sigma)$. The first expression is the same as the operational
definition of the error exponent in hypothesis testing (see, e.g.,
\cite{N06,Hay07}).

\begin{proposition}
\label{prop:sdp-err-exp-div}For states $\rho$ and $\sigma$ and $m\geq0$, the
following equalities hold%
\begin{align}
E_{d}^{m}(\rho\Vert\sigma)  &  =-\log_{2}\left[  1-\sup_{\Lambda\geq0}\left\{
\operatorname{Tr}[\Lambda\rho]:\operatorname{Tr}[\Lambda\sigma]\leq\frac
{1}{2^{m}},\ \Lambda\leq I\right\}  \right]
\label{eq:primal-err-exp-distill-dist}\\
&  =-\log_{2}\left[  1-\inf_{\lambda,W\geq0}\left\{  \frac{\lambda}{2^{m}%
}+\operatorname{Tr}[W]:\rho\leq\lambda\sigma+W\right\}  \right]  .
\label{eq:dual-err-exp-distill-dist}%
\end{align}
The complementary slackness conditions for optimal $\Lambda$, $\lambda$, and
$W$ are as follows:%
\begin{align}
\left(  \lambda\sigma+W\right)  \Lambda &  =\rho\Lambda,\\
\frac{\lambda}{2^{m}}  &  =\operatorname{Tr}[\Lambda\sigma]\lambda,\\
W  &  =\Lambda W.
\end{align}

\end{proposition}

\begin{proof}
Recalling the definition%
\begin{equation}
E_{d}^{m}(\rho\Vert\sigma)\coloneqq -\log_{2}\inf_{\mathcal{P}\in\text{CPTP}%
}\left\{  \varepsilon:\mathcal{P}(\rho)\approx_{\varepsilon}|0\rangle
\!\langle0|,\mathcal{P}(\sigma)=\pi_{2^{m^{\prime}}},m^{\prime}\geq m\right\}
,
\end{equation}
let $\mathcal{P}$ be a channel such that%
\begin{align}
\mathcal{P}(\sigma)  &  =\pi_{2^{m^{\prime}}},\\
\frac{1}{2}\left\Vert \mathcal{P}(\rho)-|0\rangle\!\langle0|\right\Vert _{1}
&  \leq\varepsilon.
\end{align}
Then it is clear that by measuring in the computational basis, i.e.,
performing the completely dephasing channel $\Delta$, it follows that%
\begin{align}
\Delta(\mathcal{P}(\sigma))  &  =\Delta(\pi_{2^{m^{\prime}}})=\pi
_{2^{m^{\prime}}},\\
\varepsilon &  \geq\frac{1}{2}\left\Vert \mathcal{P}(\rho)-|0\rangle
\langle0|\right\Vert _{1}\\
&  \geq\frac{1}{2}\left\Vert \Delta(\mathcal{P}(\rho))-\Delta(|0\rangle
\langle0|)\right\Vert _{1}\\
&  =\frac{1}{2}\left\Vert \Delta(\mathcal{P}(\rho))-|0\rangle\!\langle
0|\right\Vert _{1},
\end{align}
so that the error only decreases after doing so for all $m^{\prime}\geq m$.
Thus, it suffices to perform the optimization over quantum-to-classical
channels of the following form:%
\begin{equation}
\mathcal{M}(\omega)=\operatorname{Tr}[\Lambda\omega]|0\rangle\!\langle
0|+\operatorname{Tr}[\left(  I-\Lambda\right)  \omega]|1\rangle\!\langle1|.
\end{equation}
Then the condition $\mathcal{P}(\sigma)=\pi_{2^{m^{\prime}}}$ is equivalent to%
\begin{equation}
\operatorname{Tr}[\Lambda\sigma]=\frac{1}{2^{m^{\prime}}},
\end{equation}
and%
\begin{align}
\frac{1}{2}\left\Vert \mathcal{M}(\rho)-|0\rangle\!\langle0|\right\Vert _{1}
&  =\frac{1}{2}\left\Vert \operatorname{Tr}[\Lambda\rho]|0\rangle
\!\langle0|+\operatorname{Tr}[\left(  I-\Lambda\right)  \rho]|1\rangle
\!\langle1|-|0\rangle\!\langle0|\right\Vert _{1}\\
&  =\frac{1}{2}\left\Vert -\left(  1-\operatorname{Tr}[\Lambda\rho]\right)
|0\rangle\!\langle0|+\operatorname{Tr}[\left(  I-\Lambda\right)
\rho]|1\rangle\!\langle1|\right\Vert _{1}\\
&  =1-\operatorname{Tr}[\Lambda\rho].
\end{align}
Then the optimization above can be rewritten as%
\begin{align}
&  E_{d}^{m}(\rho\Vert\sigma)\nonumber\\
&  =-\log_{2}\inf_{\varepsilon,\Lambda\geq0}\left\{  \varepsilon
:1-\operatorname{Tr}[\Lambda\rho]\leq\varepsilon,\operatorname{Tr}%
[\Lambda\sigma]=\frac{1}{2^{m^{\prime}}},\ \Lambda\leq I,\varepsilon
\leq1,m^{\prime}\geq m\right\} \\
&  =-\log_{2}\inf_{\varepsilon,\Lambda\geq0}\left\{  \varepsilon
:1-\operatorname{Tr}[\Lambda\rho]\leq\varepsilon,\operatorname{Tr}%
[\Lambda\sigma]\leq\frac{1}{2^{m}},\ \Lambda\leq I,\varepsilon\leq1\right\}  .
\end{align}
Now consider finally that since we are trying to minimize $\varepsilon$, we
can simply set it equal to $1-\operatorname{Tr}[\Lambda\rho]$ and we finally
arrive at the following:%
\begin{align}
E_{d}^{m}(\rho\Vert\sigma)  &  =-\log_{2}\inf_{\Lambda\geq0}\left\{
1-\operatorname{Tr}[\Lambda\rho]:\operatorname{Tr}[\Lambda\sigma]\leq\frac
{1}{2^{m}},\ \Lambda\leq I\right\} \\
&  =-\log_{2}\left[  1-\sup_{\Lambda\geq0}\left\{  \operatorname{Tr}%
[\Lambda\rho]:\operatorname{Tr}[\Lambda\sigma]\leq\frac{1}{2^{m}}%
,\ \Lambda\leq I\right\}  \right]  . \label{eq:SDP-expr-err-exp}%
\end{align}
The expression inside the logarithm is thus a semi-definite program.

Using the standard form of SDPs, as stated in the appendix
\begin{align}
&  \sup_{X\geq0}\left\{  \operatorname{Tr}[AX]:\Phi(X)\leq B\right\}  ,\\
&  \inf_{Y\geq0}\left\{  \operatorname{Tr}[BY]:\Phi^{\dag}(Y)\geq A\right\}  ,
\end{align}
we can calculate the dual of \eqref{eq:SDP-expr-err-exp}. Then let us identify%
\begin{equation}
A=\rho,\qquad X=\Lambda,\qquad\Phi(X)=%
\begin{bmatrix}
\operatorname{Tr}[\Lambda\sigma] & 0\\
0 & \Lambda
\end{bmatrix}
,\qquad B=%
\begin{bmatrix}
\frac{1}{2^{m}} & 0\\
0 & I
\end{bmatrix}
. \label{eq:SDP-choices-err-exp-dist-1}%
\end{equation}
Setting%
\begin{equation}
Y=%
\begin{bmatrix}
\lambda & 0\\
0 & W
\end{bmatrix}
, \label{eq:SDP-choices-err-exp-dist-2}%
\end{equation}
we find that%
\begin{align}
\operatorname{Tr}[Y\Phi(X)]  &  =\lambda\operatorname{Tr}[\Lambda
\sigma]+\operatorname{Tr}[\Lambda W]\\
&  =\operatorname{Tr}[\left(  \lambda\sigma+W\right)  \Lambda]\\
&  =\operatorname{Tr}[\Phi^{\dag}(Y)X].
\end{align}
Thus,%
\begin{equation}
\Phi^{\dag}(Y)=\lambda\sigma+W. \label{eq:SDP-choices-err-exp-dist-3}%
\end{equation}
Then we find that the dual is given by%
\begin{equation}
\inf_{Y\geq0}\left\{  \operatorname{Tr}[BY]:\Phi^{\dag}(Y)\geq A\right\}
=\inf_{\lambda,W\geq0}\left\{  \frac{\lambda}{2^{m}}+\operatorname{Tr}%
[W]:\lambda\sigma+W\geq\rho\right\}
\end{equation}
So all of this together implies that%
\begin{align}
E_{d}^{m}(\rho\Vert\sigma)  &  =-\log_{2}\left[  1-\sup_{\Lambda\geq0}\left\{
\operatorname{Tr}[\Lambda\rho]:\operatorname{Tr}[\Lambda\sigma]\leq\frac
{1}{2^{m}},\ \Lambda\leq I\right\}  \right] \\
&  =-\log_{2}\left[  1-\inf_{\lambda,W\geq0}\left\{  \frac{\lambda}{2^{m}%
}+\operatorname{Tr}[W]:\rho\leq\lambda\sigma+W\right\}  \right]  .
\end{align}

Regarding strong duality, consider that a feasible choice for the primal is
$\Lambda=I/2^{m}$, while a strictly feasible choice for the dual is
$\lambda=1$ and $W=\left(  \rho-\sigma\right)  _{+}+I$. Thus, strong duality holds.

The complementary slackness conditions follow by examining
\eqref{eq:SDP-choices-err-exp-dist-1}--\eqref{eq:SDP-choices-err-exp-dist-3}
and \eqref{eq:comp-slack-1}--\eqref{eq:comp-slack-2}.
\end{proof}

\begin{corollary}
\label{prop:sc-reduction} The following equalities hold%
\begin{align}
\widetilde{E}_{d}^{m}(\rho\Vert\sigma)  &  =-\log_{2}\sup_{\Lambda\geq
0}\left\{  \operatorname{Tr}[\Lambda\rho]:\operatorname{Tr}[\Lambda\sigma
]\leq\frac{1}{2^{m}},\ \Lambda\leq I\right\} \\
&  =-\log_{2}\inf_{\lambda,W\geq0}\left\{  \frac{\lambda}{2^{m}}%
+\operatorname{Tr}[W]:\rho\leq\lambda\sigma+W\right\}  .
\end{align}

\end{corollary}

\begin{proof}
This follows directly from the previous result and definitions.
\end{proof}

\begin{proposition}
Let $\rho$ and $\sigma$ be states, and let $\mathcal{N}$ be a positive,
trace-preserving map. Then the following data-processing inequalities hold%
\begin{align}
E_{d}^{m}(\rho\Vert\sigma)  &  \geq E_{d}^{m}(\mathcal{N}(\rho)\Vert
\mathcal{N}(\sigma)),\label{eq:dp-err-exp}\\
\widetilde{E}_{d}^{m}(\rho\Vert\sigma)  &  \leq\widetilde{E}_{d}%
^{m}(\mathcal{N}(\rho)\Vert\mathcal{N}(\sigma)). \label{eq:dp-sc-exp}%
\end{align}

\end{proposition}

\begin{proof}
We use the primal form of $E_{d}^{m}(\rho\Vert\sigma)$ in
\eqref{eq:primal-err-exp-distill-dist}. Let $\Lambda$ be an arbitrary feasible
measurement operator for $E_{d}^{m}(\mathcal{N}(\rho)\Vert\mathcal{N}%
(\sigma))$. Then $\mathcal{N}^{\dag}(\Lambda)$ is a feasible measurement
operator for $E_{d}^{m}$. By applying definitions, the inequality in
\eqref{eq:dp-err-exp} follows. Then applying
\eqref{eq:dist-dist-err-exp-sc-exp-eq}, the inequality in \eqref{eq:dp-sc-exp} follows.
\end{proof}

\subsection{Relating distinguishability distillation to R\'enyi relative
entropies}

Let us now relate these quantities to R\'enyi relative entropies.

\begin{proposition}
\label{prop:dist-sc-exp-to-renyi}The following inequality holds%
\begin{equation}
\widetilde{E}_{d}^{m}(\rho\Vert\sigma)\geq\sup_{\alpha>1}\left(  \frac
{\alpha-1}{\alpha}\right)  \left(  m-\widetilde{D}_{\alpha}(\rho\Vert
\sigma)\right)  .
\end{equation}

\end{proposition}

\begin{proof}
This is very similar to the proof of \cite[Lemma~5]{Cooney2016}. Let $\Lambda$
be a measurement operator and suppose that $\operatorname{Tr}[\Lambda
\sigma]\leq1/2^{m}$. Then we find from data processing of the sandwiched
R\'{e}nyi relative entropy for $\alpha>1$
\cite{FL13,MO13,beigi2013sandwiched,W17f} that%
\begin{align}
&  \widetilde{D}_{\alpha}(\rho\Vert\sigma)\nonumber\\
&  \geq\frac{1}{\alpha-1}\log_{2}\left(  \operatorname{Tr}[\Lambda
\rho]^{\alpha}\operatorname{Tr}[\Lambda\sigma]^{1-\alpha}+\operatorname{Tr}%
[\left(  I-\Lambda\right)  \rho]^{\alpha}\operatorname{Tr}[\left(
I-\Lambda\right)  \sigma]^{1-\alpha}\right) \\
&  \geq\frac{1}{\alpha-1}\log_{2}\left(  \operatorname{Tr}[\Lambda
\rho]^{\alpha}\left(  1/2^{m}\right)  ^{1-\alpha}\right) \\
&  =\frac{\alpha}{\alpha-1}\log_{2}\operatorname{Tr}[\Lambda\rho]+m.
\end{align}
Rewriting this inequality, we find that%
\begin{equation}
-\log_{2}\operatorname{Tr}[\Lambda\rho]\geq\left(  \frac{\alpha-1}{\alpha
}\right)  \left(  m-\widetilde{D}_{\alpha}(\rho\Vert\sigma)\right)  .
\end{equation}
The right-hand side is independent of $\Lambda$. Since it holds for all
$\Lambda$ satisfying the given constraints, we conclude that%
\begin{equation}
\widetilde{E}_{d}^{m}(\rho\Vert\sigma)\geq\left(  \frac{\alpha-1}{\alpha
}\right)  \left(  m-\widetilde{D}_{\alpha}(\rho\Vert\sigma)\right)  .
\end{equation}
Since the right-hand side holds for all $\alpha>1$, we conclude that%
\begin{equation}
\widetilde{E}_{d}^{m}(\rho\Vert\sigma)\geq\sup_{\alpha>1}\left(  \frac
{\alpha-1}{\alpha}\right)  \left(  m-\widetilde{D}_{\alpha}(\rho\Vert
\sigma)\right)  .
\end{equation}
This concludes the proof.
\end{proof}

\bigskip Recall the following:

\begin{lemma}
[\cite{ACMBMAV07}]\label{lemma:spectral-ineq} Let $A$ and $B$ be positive
semi-definite operators, and let $s\in\left[  0,1\right]  $. Then the
following inequality holds%
\begin{equation}
\frac{1}{2}\left(  \operatorname{Tr}[A+B]-\left\Vert A-B\right\Vert
_{1}\right)  \leq\operatorname{Tr}[A^{s}B^{1-s}].
\end{equation}

\end{lemma}

\begin{proposition}
\label{prop:dist-err-exp-to-renyi}The following inequality holds%
\begin{equation}
E_{d}^{m}(\rho\Vert\sigma)\geq\sup_{\alpha\in(0,1)}\left(  \frac{\alpha
-1}{\alpha}\right)  \left(  m-D_{\alpha}(\rho\Vert\sigma)\right)
\end{equation}

\end{proposition}

\begin{proof}
This is similar to the proof of \cite[Proposition~3]{QWW17}. We exploit
Lemma~\ref{lemma:spectral-ineq} to establish the above bound. Recall from
Lemma~\ref{lemma:spectral-ineq} that the following inequality holds for
positive semi-definite operators $A$ and $B$ and for $\alpha\in(0,1)$:%
\begin{align}
\inf_{T:0\leq T\leq I}\operatorname{Tr}[(I-T)A]+\operatorname{Tr}[TB]  &
=\frac{1}{2}\left(  \operatorname{Tr}[A+B]-\left\Vert A-B\right\Vert
_{1}\right) \label{eq:aud-1}\\
&  \leq\operatorname{Tr}[A^{\alpha}B^{1-\alpha}]. \label{eq:aud-2}%
\end{align}
For the first line, see \cite[Eq.~(23)]{AM14}. For $p\in(0,1)$, pick $A=p\rho$
and $B=\left(  1-p\right)  \sigma$. Plugging in to the above inequality, we
find that there exists a measurement operator $T^{\ast}=T(p,\rho,\sigma)$ such
that%
\begin{equation}
p\operatorname{Tr}[(I-T^{\ast})\rho]+(1-p)\operatorname{Tr}[T^{\ast}%
\sigma]\leq p^{\alpha}(1-p)^{1-\alpha}\operatorname{Tr}[\rho^{\alpha}%
\sigma^{1-\alpha}].
\end{equation}
This implies that%
\begin{equation}
p\operatorname{Tr}[(I-T^{\ast})\rho]\leq p^{\alpha}(1-p)^{1-\alpha
}\operatorname{Tr}[\rho^{\alpha}\sigma^{1-\alpha}],
\end{equation}
and in turn that%
\begin{equation}
\operatorname{Tr}[(I-T^{\ast})\rho]\leq\left(  \frac{1-p}{p}\right)
^{1-\alpha}\operatorname{Tr}[\rho^{\alpha}\sigma^{1-\alpha}].
\end{equation}
Similarly, we find that%
\begin{equation}
(1-p)\operatorname{Tr}[T^{\ast}\sigma]\leq p^{\alpha}(1-p)^{1-\alpha
}\operatorname{Tr}[\rho^{\alpha}\sigma^{1-\alpha}]
\end{equation}
implies that%
\begin{equation}
\operatorname{Tr}[T^{\ast}\sigma]\leq\left(  \frac{p}{1-p}\right)  ^{\alpha
}\operatorname{Tr}[\rho^{\alpha}\sigma^{1-\alpha}].
\end{equation}
Now we pick $p\in(0,1)$ such that the following equation is satisfied%
\begin{align}
\frac{1}{2^{m}}  &  =\left(  \frac{p}{1-p}\right)  ^{\alpha}\operatorname{Tr}%
[\rho^{\alpha}\sigma^{1-\alpha}]\\
&  =\left(  \frac{p}{1-p}\right)  ^{\alpha}2^{\left(  \alpha-1\right)
D_{\alpha}(\rho\Vert\sigma)}\\
\Longleftrightarrow\frac{1}{2^{m}}2^{-\left(  \alpha-1\right)  D_{\alpha}%
(\rho\Vert\sigma)}  &  =\left(  \frac{p}{1-p}\right)  ^{\alpha}\\
\Longleftrightarrow2^{m/\alpha}2^{\left(  \alpha-1\right)  D_{\alpha}%
(\rho\Vert\sigma)/\alpha}  &  =\left(  \frac{1-p}{p}\right)  .
\end{align}
Picking $p$ in this way is possible because one more step of the development
above leads to the conclusion that%
\begin{equation}
p=\frac{1}{1+2^{m/\alpha}2^{\left(  \alpha-1\right)  D_{\alpha}(\rho
\Vert\sigma)/\alpha}}\in(0,1).
\end{equation}
Substituting above, we find that%
\begin{align}
\operatorname{Tr}[(I-T^{\ast})\rho]  &  \leq\left(  2^{m/\alpha}2^{\left(
\alpha-1\right)  D_{\alpha}(\rho\Vert\sigma)/\alpha}\right)  ^{1-\alpha
}2^{\left(  \alpha-1\right)  D_{\alpha}(\rho\Vert\sigma)}\\
&  =\left(  2^{\left(  \frac{1-\alpha}{\alpha}\right)  m}2^{-\left(
1-\alpha\right)  ^{2}D_{\alpha}(\rho\Vert\sigma)/\alpha}\right)  2^{\left(
\alpha-1\right)  D_{\alpha}(\rho\Vert\sigma)}\\
&  =2^{\left(  \frac{1-\alpha}{\alpha}\right)  m}2^{\left(  \frac{\alpha
-1}{\alpha}\right)  D_{\alpha}(\rho\Vert\sigma)}\\
&  =2^{-\left(  \frac{1-\alpha}{\alpha}\right)  \left(  D_{\alpha}(\rho
\Vert\sigma)-m\right)  }.
\end{align}
Turning this around, we have shown that there exists a measurement operator
$T^{\ast}$ such that%
\begin{align}
\operatorname{Tr}[T^{\ast}\sigma]  &  \leq\frac{1}{2^{m}},\\
-\log_{2}\left(  1-\operatorname{Tr}[T^{\ast}\rho]\right)   &  \geq\left(
\frac{1-\alpha}{\alpha}\right)  \left(  D_{\alpha}(\rho\Vert\sigma)-m\right)
.
\end{align}
Then by optimizing over all measurement operators and applying definitions, we
conclude that the following inequality holds for all $\alpha\in(0,1)$:%
\begin{equation}
E_{d}^{m}(\rho\Vert\sigma)\geq\left(  \frac{1-\alpha}{\alpha}\right)  \left(
D_{\alpha}(\rho\Vert\sigma)-m\right)  .
\end{equation}
A final optimization over $\alpha\in(0,1)$ leads to the statement of the proposition.
\end{proof}

\bigskip

Obtaining inequalities opposite to the ones above is more involved, and one
typically needs sufficiently large $n$ in order for versions of them to go
through with the correct leading order term. By making use of
Proposition~\ref{prop:sdp-err-exp-div} and Corollary~\ref{prop:sc-reduction},
we see that for $E_{d}^{m}(\rho\Vert\sigma)$, this was accomplished in
\cite{N06}, and for $\widetilde{E}_{d}^{m}(\rho\Vert\sigma)$, it was done in
\cite{MO13}.

\subsection{Relationship to smooth min-relative entropy}

Here I establish a direct relationship between the distillation error exponent
and the smooth min-relative entropy.

\begin{proposition}
\label{prop:smooth-dmin-to-dist-err-exp}Let $\rho$ and $\sigma$ be states and
let $\varepsilon\in(0,1)$. Then%
\begin{equation}
D_{\min}^{\varepsilon}(\rho\Vert\sigma)=m\qquad\Longleftrightarrow\qquad
E_{d}^{m}(\rho\Vert\sigma)=-\log_{2}\varepsilon.
\end{equation}

\end{proposition}

\begin{proof}
For fixed $\varepsilon\in\left[  0,1\right]  $, let $\Lambda$ be an optimal
measurement operator for $D_{\min}^{\varepsilon}(\rho\Vert\sigma)=m$ in
\eqref{eq:primal-SDP-smooth-min}. Then it follows that%
\begin{equation}
\operatorname{Tr}[\Lambda\sigma]=\frac{1}{2^{m}},\qquad\operatorname{Tr}%
[\left(  I-\Lambda\right)  \rho]=\varepsilon.
\end{equation}
The same measurement operator satisfies the constraints for $E_{d}^{m}%
(\rho\Vert\sigma)$ in \eqref{eq:primal-err-exp-distill-dist}\ and is thus
achievable for $E_{d}^{m}(\rho\Vert\sigma)$, implying that $E_{d}^{m}%
(\rho\Vert\sigma)\geq-\log_{2}\varepsilon$.

Now suppose that $\mu$ and $X$ are optimal for the dual formulation of
$D_{\min}^{\varepsilon}(\rho\Vert\sigma)$ in \eqref{eq:primal-SDP-smooth-min}.
Then it follows that%
\begin{equation}
\frac{1}{2^{m}}=\mu\left(  1-\varepsilon\right)  -\operatorname{Tr}[X],
\end{equation}
so that the objective function of the dual of $E_{d}^{m}(\rho\Vert\sigma)$ in
\eqref{eq:dual-err-exp-distill-dist}\ is equal to%
\begin{equation}
-\log_{2}\left(  1-\lambda\left(  \mu\left(  1-\varepsilon\right)
-\operatorname{Tr}[X]\right)  -\operatorname{Tr}[W]\right)  .
\end{equation}
Now choosing $\lambda=1/\mu$ and $W=X/\mu$, we find that these values are
feasible for the dual of $E_{d}^{m}(\rho\Vert\sigma)$, while the objective
function evaluates to $-\log_{2}\varepsilon$. So this implies that $E_{d}%
^{m}(\rho\Vert\sigma)\leq-\log_{2}\varepsilon$.

Thus, it follows that%
\begin{equation}
D_{\min}^{\varepsilon}(\rho\Vert\sigma)=m\qquad\Longrightarrow\qquad E_{d}%
^{m}(\rho\Vert\sigma)=-\log_{2}\varepsilon.
\end{equation}
To see the opposite implication, we can follow a similar method.
\end{proof}

\bigskip

The statement above is related to \cite[Eqs.~(11)--(12)]{Ren16}. It is also
stated around \cite[Eq.~(4)]{GV16}.\bigskip


The following is discussed for the classical case in \cite{Ren16}, and it has
a simple extension to the quantum case.

\begin{proposition}
Let $\varepsilon\in(0,1)$ and set $m\coloneqq \log_{2}\left(  \frac
{1}{\varepsilon}\right)  $. Then the following identity holds%
\begin{equation}
2^{-\widetilde{E}_{d}^{m}(\sigma\Vert\rho)}+2^{-D_{\min}^{\varepsilon}%
(\rho\Vert\sigma)}=1. \label{eq:smooth-dmin-to-sc-exp}%
\end{equation}
This is equivalent to the following:%
\begin{equation}
D_{\min}^{\varepsilon}(\rho\Vert\sigma)=E_{d}^{m}(\sigma\Vert\rho).
\label{eq:smooth-dmin-to-err-exp}%
\end{equation}

\end{proposition}

\begin{proof}
This identity is essentially the same as that given in \cite[Eq.~(13)]{Ren16},
and it is a direct consequence of definitions. Note that%
\begin{align}
2^{-\widetilde{E}_{d}^{m}(\sigma\Vert\rho)}  &  =\sup_{\Lambda\geq0}\left\{
\operatorname{Tr}[\Lambda\sigma]:\operatorname{Tr}[\Lambda\rho]\leq
\varepsilon,\ \Lambda\leq I\right\}  ,\\
2^{-D_{\min}^{\varepsilon}(\rho\Vert\sigma)}  &  =\inf_{\Lambda\geq0}\left\{
\operatorname{Tr}[\Lambda\sigma]:\operatorname{Tr}[\Lambda\rho]\geq
1-\varepsilon,\ \Lambda\leq I\right\}  .
\end{align}
Thus it follows that%
\begin{align}
1-2^{-\widetilde{E}_{d}^{m}(\sigma\Vert\rho)}  &  =1-\sup_{\Lambda\geq
0}\left\{  \operatorname{Tr}[\Lambda\sigma]:\operatorname{Tr}[\Lambda\rho
]\leq\varepsilon,\ \Lambda\leq I\right\} \\
&  =\inf_{\Lambda\geq0}\left\{  1-\operatorname{Tr}[\Lambda\sigma
]:\operatorname{Tr}[\Lambda\rho]\leq\varepsilon,\ \Lambda\leq I\right\} \\
&  =\inf_{\Lambda\geq0}\left\{  \operatorname{Tr}[\left(  I-\Lambda\right)
\sigma]:\operatorname{Tr}[\Lambda\rho]\leq\varepsilon,\ \Lambda\leq I\right\}
\\
&  =\inf_{\Lambda\geq0}\left\{  \operatorname{Tr}[\Lambda\sigma
]:\operatorname{Tr}[\left(  I-\Lambda\right)  \rho]\leq\varepsilon
,\ \Lambda\leq I\right\} \\
&  =2^{-D_{\min}^{\varepsilon}(\rho\Vert\sigma)}.
\end{align}

The equality in \eqref{eq:smooth-dmin-to-err-exp} follows from
\eqref{eq:smooth-dmin-to-sc-exp} and \eqref{eq:dist-dist-err-exp-sc-exp-eq}.
\end{proof}

\section{Distinguishability dilution}

\subsection{Semi-definite programs}

\begin{proposition}
\label{prop:dist-dil-sdp} Let $\rho$ and $\sigma$ be states and $m\geq0$. The
following equality holds%
\begin{align}
E_{c}^{m}(\rho\Vert\sigma)  &  =-\log_{2}\inf_{Z,\widetilde{\rho}\geq
0}\left\{  \operatorname{Tr}[Z]:Z\geq\widetilde{\rho}-\rho,\ \widetilde{\rho
}\leq2^{m}\sigma,\ \operatorname{Tr}[\widetilde{\rho}]=1\right\}
\label{eq:primal-cost-err-exp}\\
&  =-\log_{2}\sup_{\kappa,R,S\geq0}\left\{  \kappa-\operatorname{Tr}%
[R\rho]-2^{m}\operatorname{Tr}[S\sigma]:R\leq I,\ \kappa I\leq R+S\right\}  .
\label{eq:dual-cost-err-exp}%
\end{align}
The complementary slackness conditions for optimal $Z$, $\widetilde{\rho}$,
$\kappa$, $R$, and $S$ are as follows:%
\begin{align}
Z  &  =RZ,\\
\kappa\widetilde{\rho}  &  =\left(  R+S\right)  \widetilde{\rho},\\
ZR  &  =\left(  \widetilde{\rho}-\rho\right)  R,\\
2^{m}\sigma S  &  =\widetilde{\rho}S.
\end{align}
The first quantity in \eqref{eq:primal-cost-err-exp}\ can be written as%
\begin{equation}
E_{c}^{m}(\rho\Vert\sigma)=-\log_{2}\inf_{\widetilde{\rho}\in\mathcal{S}%
,\ \widetilde{\rho}\leq2^{m}\sigma}\frac{1}{2}\left\Vert \widetilde{\rho}%
-\rho\right\Vert _{1},
\end{equation}
where $\mathcal{S}$ denotes the set of density operators.
\end{proposition}

\begin{proof}
We begin by determining how to write the quantity $E_{c}^{m}(\rho\Vert\sigma)$
as a semi-definite program. Recall that%
\begin{equation}
E_{c}^{m}(\rho\Vert\sigma)\coloneqq -\log_{2}\inf_{\mathcal{P}\in\text{CPTP}%
}\left\{  \varepsilon:\mathcal{P}(|0\rangle\!\langle0|)\approx_{\varepsilon
}\rho,\mathcal{P}(\pi_{2^{m^{\prime}}})=\sigma,\ m^{\prime}\leq m\right\}  .
\end{equation}
In this case, since we are starting from classical states, it suffices for
$\mathcal{P}$ to be a classical--quantum channel of the following form:%
\begin{equation}
\mathcal{P}(\omega)=\langle0|\omega|0\rangle\widetilde{\rho}+\langle
1|\omega|1\rangle\tau.
\end{equation}
The equality constraint $\mathcal{P}(\pi_{2^{m^{\prime}}})=\sigma$ implies
that%
\begin{align}
\sigma &  =\langle0|\left(  \frac{1}{2^{m^{\prime}}}|0\rangle\!\langle
0|+\left(  1-\frac{1}{2^{m^{\prime}}}\right)  |1\rangle\!\langle1|\right)
|0\rangle\widetilde{\rho}\nonumber\\
&  \qquad+\langle1|\left(  \frac{1}{2^{m^{\prime}}}|0\rangle\!\langle
0|+\left(  1-\frac{1}{2^{m^{\prime}}}\right)  |1\rangle\!\langle1|\right)
|1\rangle\tau\\
&  =\frac{1}{2^{m^{\prime}}}\widetilde{\rho}+\left(  1-\frac{1}{2^{m^{\prime}%
}}\right)  \tau.
\end{align}
Consider that this implies that%
\begin{equation}
\widetilde{\rho}=2^{m^{\prime}}\sigma-\left(  2^{m^{\prime}}-1\right)  \tau.
\end{equation}
The constraint $\mathcal{P}(|0\rangle\!\langle0|)\approx_{\varepsilon}\rho$
implies that%
\begin{equation}
\frac{1}{2}\left\Vert \widetilde{\rho}-\rho\right\Vert _{1}\leq\varepsilon.
\end{equation}
Now considering that%
\begin{equation}
\frac{1}{2}\left\Vert \widetilde{\rho}-\rho\right\Vert _{1}=\inf_{Z\geq
0}\left\{  \operatorname{Tr}[Z]:Z\geq\widetilde{\rho}-\rho\right\}  ,
\end{equation}
we find that%
\begin{align}
&  E_{c}^{m}(\rho\Vert\sigma)\nonumber\\
&  =-\log_{2}\inf_{\varepsilon\in\left[  0,1\right]  ,Z,\widetilde{\rho}%
,\tau\geq0}\left\{
\begin{array}
[c]{c}%
\varepsilon:\varepsilon\geq\operatorname{Tr}[Z],Z\geq\widetilde{\rho}%
-\rho,\widetilde{\rho}=2^{m^{\prime}}\sigma-\left(  2^{m^{\prime}}-1\right)
\tau,\\
m^{\prime}\leq m,\operatorname{Tr}[\widetilde{\rho}]=1,\operatorname{Tr}%
[\tau]=1.
\end{array}
\right\} \\
&  =-\log_{2}\inf_{Z,\widetilde{\rho},\tau\geq0}\left\{
\begin{array}
[c]{c}%
\operatorname{Tr}[Z]:Z\geq\widetilde{\rho}-\rho,\widetilde{\rho}=2^{m^{\prime
}}\sigma-\left(  2^{m^{\prime}}-1\right)  \tau,\\
m^{\prime}\leq m,\operatorname{Tr}[\widetilde{\rho}]=1,\operatorname{Tr}%
[\tau]=1.
\end{array}
\right\}  .
\end{align}
Now consider that the condition that there exists a state $\tau$ satisfying
\begin{equation}
\widetilde{\rho}=2^{m^{\prime}}\sigma-\left(  2^{m^{\prime}}-1\right)  \tau
\end{equation}
is equivalent to the condition%
\begin{equation}
\widetilde{\rho}\leq2^{m^{\prime}}\sigma.
\end{equation}
If the first condition is true, then the second one clearly is. If the second
condition is true, then we can set%
\begin{equation}
\tau=\frac{2^{m^{\prime}}\sigma-\rho}{2^{m^{\prime}}-1},
\end{equation}
which is a legitimate state. So we get the further simplification:%
\begin{align}
E_{c}^{m}(\rho\Vert\sigma)  &  =-\log_{2}\inf_{Z,\widetilde{\rho}\geq
0}\left\{
\begin{array}
[c]{c}%
\operatorname{Tr}[Z]:Z\geq\widetilde{\rho}-\rho,\ \widetilde{\rho}%
\leq2^{m^{\prime}}\sigma,\\
m^{\prime}\leq m,\operatorname{Tr}[\widetilde{\rho}]=1.
\end{array}
\right\} \\
&  =-\log_{2}\inf_{Z,\widetilde{\rho}\geq0}\left\{  \operatorname{Tr}%
[Z]:Z\geq\widetilde{\rho}-\rho,\ \widetilde{\rho}\leq2^{m}\sigma
,\ \operatorname{Tr}[\widetilde{\rho}]=1\right\}  .
\end{align}

We can use the standard form of SDPs to compute the dual:%
\begin{align}
&  \sup_{X\geq0}\left\{  \operatorname{Tr}[AX]:\Phi(X)\leq B\right\}  ,\\
&  \inf_{Y\geq0}\left\{  \operatorname{Tr}[BY]:\Phi^{\dag}(Y)\geq A\right\}  .
\end{align}
Here we identify%
\begin{align}
Y  &  =%
\begin{bmatrix}
Z & 0\\
0 & \widetilde{\rho}%
\end{bmatrix}
,\qquad B=%
\begin{bmatrix}
I & 0\\
0 & 0
\end{bmatrix}
,\label{eq:err-exp-dist-cost-1}\\
\Phi^{\dag}(Y)  &  =\text{diag}(Z-\widetilde{\rho},-\widetilde{\rho
},\operatorname{Tr}[\widetilde{\rho}],-\operatorname{Tr}[\widetilde{\rho}]),\\
A  &  =\text{diag}(-\rho,-2^{m}\sigma,1,-1). \label{eq:err-exp-dist-cost-3}%
\end{align}
Now setting%
\begin{equation}
X=\text{diag}\left(  R,S,\kappa_{1},\kappa_{2}\right)  ,
\label{eq:err-exp-dist-cost-4}%
\end{equation}
we find that%
\begin{align}
\operatorname{Tr}[\Phi^{\dag}(Y)X]  &  =\operatorname{Tr}[\left(
Z-\widetilde{\rho}\right)  R]-\operatorname{Tr}[\widetilde{\rho}S]+\left(
\kappa_{1}-\kappa_{2}\right)  \operatorname{Tr}[\widetilde{\rho}]\\
&  =\operatorname{Tr}[ZR]+\operatorname{Tr}[\widetilde{\rho}\left(  \left[
\kappa_{1}-\kappa_{2}\right]  I-R-S\right)  ]\\
&  =\operatorname{Tr}\left[
\begin{bmatrix}
Z & 0\\
0 & \widetilde{\rho}%
\end{bmatrix}%
\begin{bmatrix}
R & 0\\
0 & \left[  \kappa_{1}-\kappa_{2}\right]  I-R-S
\end{bmatrix}
\right] \\
&  =\operatorname{Tr}[Y\Phi(X)].
\end{align}
So we conclude that%
\begin{equation}
\Phi(X)=%
\begin{bmatrix}
R & 0\\
0 & \left[  \kappa_{1}-\kappa_{2}\right]  I-R-S
\end{bmatrix}
. \label{eq:err-exp-dist-cost-5}%
\end{equation}
Then the dual is given by%
\begin{align}
&  \sup_{X\geq0}\left\{  \operatorname{Tr}[AX]:\Phi(X)\leq B\right\} \notag \\
&  =\sup_{R,S,\kappa_{1},\kappa_{2}\geq0}\left\{  \kappa_{1}-\kappa
_{2}-\operatorname{Tr}[R\rho]-2^{m}\operatorname{Tr}[S\sigma]:R\leq
I,\ \left[  \kappa_{1}-\kappa_{2}\right]  I-R-S\leq0\right\}  .
\end{align}
This can be rewritten as follows:%
\begin{equation}
\sup_{R,S\geq0,\kappa\in\mathbb{R}}\left\{  \kappa-\operatorname{Tr}%
[R\rho]-2^{m}\operatorname{Tr}[S\sigma]:R\leq I,\ \kappa I\leq R+S\right\}  .
\end{equation}
Since we know that the minimum value of the primal is $\geq0$, we can then
optimize exclusively over $\kappa\geq0$ in the dual. The final form is as
follows:%
\begin{equation}
\sup_{\kappa,R,S\geq0}\left\{  \kappa-\operatorname{Tr}[R\rho]-2^{m}%
\operatorname{Tr}[S\sigma]:R\leq I,\ \kappa I\leq R+S\right\}  .
\end{equation}

Strong duality holds because the values $\widetilde{\rho}=\sigma$ and
$Z=\left(  \sigma-\rho\right)  _{+}$ are feasible for the primal, while the
values $R=I/2$, $S=I$, and $\kappa=1/2$ are strictly feasible for the dual.
The complementary slackness conditions follow by applying
\eqref{eq:comp-slack-1}--\eqref{eq:comp-slack-2} to
\eqref{eq:err-exp-dist-cost-1}--\eqref{eq:err-exp-dist-cost-4} and \eqref{eq:err-exp-dist-cost-5}.
\end{proof}

\begin{corollary}
\label{cor:sc-exp-dist-dil}The following equalities hold%
\begin{align}
&  \widetilde{E}_{c}^{m}(\rho\Vert\sigma)\nonumber\\
&  =-\log_{2}\left[  1-\inf_{Z,\widetilde{\rho}\geq0}\left\{
\operatorname{Tr}[Z]:Z\geq\widetilde{\rho}-\rho,\ \widetilde{\rho}\leq
2^{m}\sigma,\ \operatorname{Tr}[\widetilde{\rho}]=1\right\}  \right] \\
&  =-\log_{2}\left[  1-\sup_{\kappa,R,S\geq0}\left\{  \kappa-\operatorname{Tr}%
[R\rho]-2^{m}\operatorname{Tr}[S\sigma]:R\leq I,\ \kappa I\leq R+S\right\}
\right]  .
\end{align}

\end{corollary}

\begin{proof}
Direct consequence of definitions and the previous proposition.
\end{proof}

\begin{proposition}
Let $\rho$ and $\sigma$ be states, and let $\mathcal{N}$ be a positive,
trace-preserving map. Then the following data-processing inequalities hold%
\begin{align}
E_{c}^{m}(\rho\Vert\sigma)  &  \geq E_{c}^{m}(\mathcal{N}(\rho)\Vert
\mathcal{N}(\sigma)),\label{eq:dp-dist-dilu-exp}\\
\widetilde{E}_{c}^{m}(\rho\Vert\sigma)  &  \leq\widetilde{E}_{c}%
^{m}(\mathcal{N}(\rho)\Vert\mathcal{N}(\sigma)).
\label{eq:dp-dist-dilu-sc-exp}%
\end{align}

\end{proposition}

\begin{proof}
We use the dual form of $E_{c}^{m}(\rho\Vert\sigma)$ in
\eqref{eq:dual-cost-err-exp}. Let $R$, $S$, and $\kappa$ be feasible choices
for $\mathcal{N}(\rho)$ and $\mathcal{N}(\sigma)$ in $E_{c}^{m}(\mathcal{N}%
(\rho)\Vert\mathcal{N}(\sigma))$. Then $\mathcal{N}^{\dag}(R)$, $\mathcal{N}%
^{\dag}(S)$, and $\kappa$ are feasible choices for $E_{c}^{m}(\rho\Vert
\sigma)$. By applying definitions, the inequality in
\eqref{eq:dp-dist-dilu-exp} follows. Applying
\eqref{eq:relate-err-exp-sc-exp-cost}, the inequality in \eqref{eq:dp-dist-dilu-sc-exp}\ follows.
\end{proof}

\subsection{Relating distinguishability dilution to R\'enyi relative
entropies}

Using the above, the strong converse exponent can then be written as%
\begin{align}
\widetilde{E}_{c}^{m}(\rho\Vert\sigma)  &  \coloneqq-\log_{2}\left[
1-\inf_{\widetilde{\rho}\in\mathcal{S},\widetilde{\rho}\leq2^{m}\sigma}%
\frac{1}{2}\left\Vert \widetilde{\rho}-\rho\right\Vert _{1}\right] \\
&  =\inf_{\widetilde{\rho}\in\mathcal{S},\widetilde{\rho}\leq2^{m}\sigma
}\left(  -\log_{2}\left[  1-\frac{1}{2}\left\Vert \widetilde{\rho}%
-\rho\right\Vert _{1}\right]  \right)  . \label{eq:alt-form-sc-exp-cost}%
\end{align}

\begin{proposition}
\label{prop:one-shot-lower-bnd-dist-dil}For states $\rho$ and $\sigma$ and
$m\geq0$, the following inequality holds%
\begin{multline}
\max\left\{  \sup_{\alpha\in\left(  0,1\right)  }\left(  \frac{\alpha-1}%
{2}\right)  \left[  m-D_{\alpha}(\rho\Vert\sigma)\right]  ,\sup_{\alpha
\in\left(  1/2,1\right)  }\left(  \frac{\alpha-1}{2\alpha}\right)  \left[
m-\widetilde{D}_{\alpha}(\rho\Vert\sigma)\right]  \right\} \\
\leq\widetilde{E}_{c}^{m}(\rho\Vert\sigma).
\end{multline}

\end{proposition}

\begin{proof}
Recall the following inequality from \cite[Lemma~3]{Wang2019states}:%
\begin{equation}
D_{\beta}(\rho_{0}\Vert\sigma)-D_{\alpha}(\rho_{1}\Vert\sigma)\geq\frac
{2}{1-\alpha}\log_{2}\left[  1-\frac{1}{2}\left\Vert \rho_{0}-\rho
_{1}\right\Vert _{1}\right]  ,
\end{equation}
which holds for $\alpha\in(0,1)$ and $\beta=2-\alpha\in(1,2)$. We can rewrite
this as follows:%
\begin{equation}
\left(  \frac{1-\alpha}{2}\right)  \left[  D_{\alpha}(\rho_{1}\Vert
\sigma)-D_{\beta}(\rho_{0}\Vert\sigma)\right]  \leq-\log_{2}\left[  1-\frac
{1}{2}\left\Vert \rho_{0}-\rho_{1}\right\Vert _{1}\right]  .
\end{equation}
Fix $\alpha\in(0,1)$. Let $\widetilde{\rho}$ be an arbitrary state satisfying
$\widetilde{\rho}\leq2^{m}\sigma$. Then we find that%
\begin{equation}
\left(  \frac{1-\alpha}{2}\right)  \left[  D_{\alpha}(\rho\Vert\sigma
)-D_{\beta}(\widetilde{\rho}\Vert\sigma)\right]  \leq-\log_{2}\left[
1-\frac{1}{2}\left\Vert \widetilde{\rho}-\rho\right\Vert _{1}\right]  .
\end{equation}
Now consider that%
\begin{align}
D_{\beta}(\widetilde{\rho}\Vert\sigma)  &  =\frac{1}{\beta-1}\log
_{2}\operatorname{Tr}[\widetilde{\rho}^{\beta}\sigma^{1-\beta}]\\
&  \leq\frac{1}{\beta-1}\log_{2}\operatorname{Tr}[\widetilde{\rho}^{\beta
}\left(  2^{-m}\widetilde{\rho}\right)  ^{1-\beta}]\\
&  =m+\frac{1}{\beta-1}\log_{2}\operatorname{Tr}[\widetilde{\rho}^{\beta
}\widetilde{\rho}^{1-\beta}]\\
&  =m,
\end{align}
which follows from operator anti-monotonicity of $x^{1-\beta}$ for $\beta
\in(1,2)$. Equivalently, we could also simply use the fact that $D_{\beta
}(\widetilde{\rho}\Vert\sigma)\leq D_{\max}(\widetilde{\rho}\Vert\sigma)\leq
m$ \cite{BD10}. Substituting this above, we find that%
\begin{equation}
\left(  \frac{1-\alpha}{2}\right)  \left[  D_{\alpha}(\rho\Vert\sigma
)-m\right]  \leq-\log_{2}\left[  1-\frac{1}{2}\left\Vert \widetilde{\rho}%
-\rho\right\Vert _{1}\right]  .
\end{equation}
Since the inequality holds for all states $\widetilde{\rho}$ satisfying
$\widetilde{\rho}\leq2^{m}\sigma$, we conclude that%
\begin{equation}
\left(  \frac{1-\alpha}{2}\right)  \left[  D_{\alpha}(\rho\Vert\sigma
)-m\right]  \leq\widetilde{E}_{c}^{m}(\rho\Vert\sigma).
\end{equation}
Since the inequality holds for all $\alpha\in(0,1)$, we conclude the first
statement of the proposition.

Recall that \cite[Lemma~3]{Wang2019states}%
\begin{equation}
\widetilde{D}_{\beta}(\rho_{0}\Vert\sigma)-\widetilde{D}_{\alpha}(\rho
_{1}\Vert\sigma)\geq\frac{\alpha}{1-\alpha}\log_{2}F(\rho_{0},\rho_{1}),
\end{equation}
for $\alpha\in\left(  1/2,1\right)  $ and $\beta(\alpha)\coloneqq\alpha
/(2\alpha-1)\in(1,\infty)$. We can rewrite this as%
\begin{align}
&  \left(  \frac{1-\alpha}{2\alpha}\right)  \left[  \widetilde{D}_{\alpha
}(\rho_{1}\Vert\sigma)-\widetilde{D}_{\beta}(\rho_{0}\Vert\sigma)\right]
\nonumber\\
&  \leq-\frac{1}{2}\log_{2}F(\rho_{0},\rho_{1})\\
&  =-\log_{2}\sqrt{F}(\rho_{0},\rho_{1})\\
&  \leq-\log_{2}\left[  1-\frac{1}{2}\left\Vert \rho_{0}-\rho_{1}\right\Vert
_{1}\right]  ,
\end{align}
where we made use of the inequality $\sqrt{F}(\rho_{0},\rho_{1})\geq1-\frac
{1}{2}\left\Vert \rho_{0}-\rho_{1}\right\Vert _{1}$ \cite{FG98}. Fix
$\alpha\in\left(  1/2,1\right)  $. Let $\widetilde{\rho}$ be an arbitrary
state satisfying $\widetilde{\rho}\leq2^{m}\sigma$. Then we find that%
\begin{equation}
\left(  \frac{1-\alpha}{2\alpha}\right)  \left[  \widetilde{D}_{\alpha}%
(\rho\Vert\sigma)-\widetilde{D}_{\beta}(\widetilde{\rho}\Vert\sigma)\right]
\leq-\log_{2}\left[  1-\frac{1}{2}\left\Vert \widetilde{\rho}-\rho\right\Vert
_{1}\right]  .
\end{equation}
Now consider that \cite{muller2013quantum}%
\begin{equation}
\widetilde{D}_{\beta}(\widetilde{\rho}\Vert\sigma)\leq D_{\max}(\widetilde
{\rho}\Vert\sigma)\leq m.
\end{equation}
This implies that%
\begin{equation}
\left(  \frac{1-\alpha}{2\alpha}\right)  \left[  \widetilde{D}_{\alpha}%
(\rho\Vert\sigma)-m\right]  \leq-\log_{2}\left[  1-\frac{1}{2}\left\Vert
\widetilde{\rho}-\rho\right\Vert _{1}\right]  .
\end{equation}
Since the inequality holds for all states $\widetilde{\rho}$ satisfying
$\widetilde{\rho}\leq2^{m}\sigma$, we conclude that%
\begin{equation}
\left(  \frac{1-\alpha}{2\alpha}\right)  \left[  \widetilde{D}_{\alpha}%
(\rho\Vert\sigma)-m\right]  \leq\widetilde{E}_{c}^{m}(\rho\Vert\sigma).
\end{equation}
Since the inequality holds for all $\alpha\in\left(  1/2,1\right)  $, we
conclude the second statement of the proposition.
\end{proof}

\begin{corollary}
\label{cor:sc-exp-dd-cost} Let $\rho$ and $\sigma$ be states. Then the
following lower bound holds for the asymptotic strong converse exponent of
distinguishability dilution:%
\begin{multline}
\max\left\{  \sup_{\alpha\in\left(  0,1\right)  }\left(  \frac{\alpha-1}%
{2}\right)  \left[  R-D_{\alpha}(\rho\Vert\sigma)\right]  ,\sup_{\alpha
\in\left(  1/2,1\right)  }\left(  \frac{1-\alpha}{2\alpha}\right)  \left[
R-\widetilde{D}_{\alpha}(\rho\Vert\sigma)\right]  \right\}
\label{eq:lower-bnd-sc-exp-dist-dil}\\
\leq\lim_{n\rightarrow\infty}\frac{1}{n}\widetilde{E}_{c}^{nR}(\rho^{\otimes
n}\Vert\sigma^{\otimes n}).
\end{multline}

\end{corollary}

\begin{proof}
This follows from definitions and a direct application of
Proposition~\ref{prop:one-shot-lower-bnd-dist-dil}.
\end{proof}

\begin{proposition}
\label{prop:err-exp-cost-to-renyi}Let $\rho$ and $\sigma$ be states and
$m\geq0$. Then the following bound holds for all $\alpha>1$:%
\begin{equation}
\left(  \frac{\alpha-1}{2}\right)  \left[  m-\widetilde{D}_{\alpha}(\rho
\Vert\sigma)\right]  \leq E_{c}^{m}(\rho\Vert\sigma)+\left(  \frac{\alpha
-1}{2}\right)  \log_{2}\left(  \frac{1}{1-2^{-2E_{c}^{m}(\rho\Vert\sigma)}%
}\right)
\end{equation}

\end{proposition}

\begin{proof}
Consider from \cite[Proposition~6]{Wang2019states} and \cite[Proposition~2.2]%
{Buscemi2019information} that the following inequality holds for $\alpha>1$
and $\varepsilon\in(0,1)$:%
\begin{equation}
D_{\max}^{\varepsilon}(\rho\Vert\sigma)\leq\widetilde{D}_{\alpha}(\rho
\Vert\sigma)+\frac{1}{\alpha-1}\log_{2}\left(  \frac{1}{\varepsilon^{2}%
}\right)  +\log_{2}\left(  \frac{1}{1-\varepsilon^{2}}\right)  .
\end{equation}
Setting $m=D_{\max}^{\varepsilon}(\rho\Vert\sigma)$ so that $E_{c}^{m}%
(\rho\Vert\sigma)=-\log_{2}(\varepsilon)$ by
Proposition~\ref{prop:smooth-dmax-to-err-exp-cost}, and%
\begin{equation}
\log_{2}\left(  \frac{1}{1-\varepsilon^{2}}\right)  =\log_{2}\left(  \frac
{1}{1-2^{-2E_{c}^{m}(\rho\Vert\sigma)}}\right)
\end{equation}
we find that%
\begin{equation}
m\leq\widetilde{D}_{\alpha}(\rho\Vert\sigma)+\frac{2}{\alpha-1}E_{c}^{m}%
(\rho\Vert\sigma)+\log_{2}\left(  \frac{1}{1-2^{-2E_{c}^{m}(\rho\Vert\sigma)}%
}\right)  ,
\end{equation}
which implies that%
\begin{equation}
\left(  \frac{\alpha-1}{2}\right)  \left[  m-\widetilde{D}_{\alpha}(\rho
\Vert\sigma)\right]  \leq E_{c}^{m}(\rho\Vert\sigma)+\left(  \frac{\alpha
-1}{2}\right)  \log_{2}\left(  \frac{1}{1-2^{-2E_{c}^{m}(\rho\Vert\sigma)}%
}\right)  .
\end{equation}
This concludes the proof.
\end{proof}

\subsection{Relationship to smooth max-relative entropy}

\begin{proposition}
\label{prop:smooth-dmax-to-err-exp-cost}Let $\rho$ and $\sigma$ be states and
let $\varepsilon\in(0,1)$. Then%
\begin{equation}
D_{\max}^{\varepsilon}(\rho\Vert\sigma)=m\qquad\Longleftrightarrow\qquad
E_{c}^{m}(\rho\Vert\sigma)=-\log_{2}\varepsilon.
\end{equation}

\end{proposition}

\begin{proof}
Suppose that $D_{\max}^{\varepsilon}(\rho\Vert\sigma)=m$. By applying
\eqref{eq:smooth-dmax-SDP}, we conclude that there exists a state
$\widetilde{\rho}$ and $\lambda\geq0$ satisfying $\widetilde{\rho}\leq
\lambda\sigma$ and an operator $Z\geq0$ satisfying $Z\geq\widetilde{\rho}%
-\rho$ and $\operatorname{Tr}[Z]\leq\varepsilon$, such that $2^{m}=\lambda$.
These choices are then feasible for $E_{c}^{m}(\rho\Vert\sigma)$ as given in
\eqref{eq:primal-cost-err-exp}, and so we conclude that%
\begin{equation}
E_{c}^{m}(\rho\Vert\sigma)\geq-\log_{2}\varepsilon.
\label{eq:err-exp-cost-lower-pf-1}%
\end{equation}

Let $t$, $X$, $Q\geq0$ and $\mu\in\mathbb{R}$ be optimal for the dual
formulation of $D_{\max}^{\varepsilon}(\rho\Vert\sigma)$, so that they satisfy%
\begin{equation}
\operatorname{Tr}[X\sigma]=1,\qquad Q\leq tI,\qquad Q+\mu I\leq X,
\end{equation}
as well as $2^{m}=\operatorname{Tr}[Q\rho]+\mu-\varepsilon t$. Consider that
the objective function of $E_{c}^{m}(\rho\Vert\sigma)$ is%
\begin{equation}
\kappa-\operatorname{Tr}[R\rho]-2^{m}\operatorname{Tr}[S\sigma]
\end{equation}
for which the following constraints hold%
\begin{equation}
\kappa,R,S\geq0,\qquad R\leq I,\qquad\kappa I\leq R+S.
\end{equation}
Let us pick $S=X/t$, $R=I-Q/t$, and $\kappa=1+\mu/t$, and it follows that the
needed constraints hold. We then find that the objective function evaluates to%
\begin{align}
&  \kappa-\operatorname{Tr}[R\rho]-2^{m}\operatorname{Tr}[S\sigma]\nonumber\\
&  =1+\mu/t-\operatorname{Tr}[\left(  I-Q/t\right)  \rho]-\left(
\operatorname{Tr}[Q\rho]+\mu-\varepsilon t\right)  \operatorname{Tr}%
[X\sigma]/t\\
&  =1+\mu/t-1+\operatorname{Tr}[Q\rho]/t-\left(  \operatorname{Tr}[Q\rho
]+\mu-\varepsilon t\right)  /t\\
&  =1+\mu/t-1+\operatorname{Tr}[Q\rho]/t-\left(  \operatorname{Tr}%
[Q\rho]/t+\mu/t-\varepsilon\right) \\
&  =\varepsilon.
\end{align}
By applying definitions, it follows that%
\begin{equation}
E_{c}^{m}(\rho\Vert\sigma)\leq-\log_{2}\varepsilon.
\label{eq:err-exp-cost-upper-pf-2}%
\end{equation}
Combining \eqref{eq:err-exp-cost-lower-pf-1} and
\eqref{eq:err-exp-cost-upper-pf-2}, we conclude that $E_{c}^{m}(\rho
\Vert\sigma)=-\log_{2}\varepsilon$.

We can show the opposite implication by inverting the choices above. Starting
from the optimal choices in the dual of $E_{c}^{m}(\rho\Vert\sigma)$, choose
$X=St$, $t=1/\operatorname{Tr}[S\sigma]$, $Q=\left(  I-R\right)  t$,
$\mu=\left(  \kappa-1\right)  t$. Then we find that the constraints for the
dual of $D_{\max}^{\varepsilon}(\rho\Vert\sigma)$ are satisfied and that
$\operatorname{Tr}[Q\rho]+\mu-\varepsilon t=2^{m}$. It then follows that
$D_{\max}^{\varepsilon}(\rho\Vert\sigma)\geq m$. Similarly, from the optimal
choices of the primal of $E_{c}^{m}(\rho\Vert\sigma)$, we find that $D_{\max
}^{\varepsilon}(\rho\Vert\sigma)\leq m$. So we conclude the other implication.
\end{proof}

\section{Relating distinguishability distillation and dilution exponents}

\begin{proposition}
\label{prop:bounds-taken-from-smooth-dmin-dmax}Let $\rho$ and $\sigma$ be
states and let $k,m\geq0$. Then the following inequalities hold%
\begin{align}
-\log_{2}\left(  2^{-E_{c}^{k}(\rho\Vert\sigma)}+2^{k-m}\right)   &
\leq\widetilde{E}_{d}^{m}(\rho\Vert\sigma),\\
-\log_{2}\left(  2^{-E_{d}^{m}(\rho\Vert\sigma)}+2^{k-m}\right)   &
\leq\widetilde{E}_{c}^{k}(\rho\Vert\sigma).
\end{align}

\end{proposition}

\begin{proof}
The following inequality is known \cite{Wang2019states}:%
\begin{equation}
D_{\min}^{\varepsilon_{1}}(\rho\Vert\sigma)\leq D_{\max}^{\varepsilon_{2}%
}(\rho\Vert\sigma)+\log_{2}\left(  \frac{1}{1-\varepsilon_{1}-\varepsilon_{2}%
}\right)  .
\end{equation}
By exploiting the identities in
Propositions~\ref{prop:smooth-dmin-to-dist-err-exp}\ and
\ref{prop:smooth-dmax-to-err-exp-cost}, and setting $D_{\min}^{\varepsilon
_{1}}(\rho\Vert\sigma)=m$ and $D_{\max}^{\varepsilon_{2}}(\rho\Vert\sigma)=k$,
while noticing that%
\begin{align}
E_{d}^{m}(\rho\Vert\sigma)  &  =-\log_{2}(\varepsilon_{1}),\\
E_{c}^{k}(\rho\Vert\sigma)  &  =-\log_{2}(\varepsilon_{2}),
\end{align}
we find that the inequality above translates to%
\begin{multline}
m\leq k+\log_{2}\left(  \frac{1}{1-2^{-E_{d}^{m}(\rho\Vert\sigma)}%
-2^{-E_{c}^{k}(\rho\Vert\sigma)}}\right) \\
\qquad\Longleftrightarrow\qquad1-2^{-E_{d}^{m}(\rho\Vert\sigma)}-2^{-E_{c}%
^{k}(\rho\Vert\sigma)}\leq2^{k-m}.
\end{multline}
This latter inequality implies the following inequality:%
\begin{align}
2^{-E_{c}^{k}(\rho\Vert\sigma)}+2^{k-m}  &  \geq1-2^{-E_{d}^{m}(\rho
\Vert\sigma)}\\
&  =2^{-\widetilde{E}_{d}^{m}(\rho\Vert\sigma)}.
\end{align}
It also implies the following inequality:%
\begin{align}
2^{-E_{d}^{m}(\rho\Vert\sigma)}+2^{k-m}  &  \geq1-2^{-E_{c}^{k}(\rho
\Vert\sigma)}\\
&  =2^{-\widetilde{E}_{c}^{k}(\rho\Vert\sigma)}.
\end{align}
By adding one to each side, we find that%
\begin{align}
1+2^{k-m}  &  \geq1-2^{-E_{d}^{m}(\rho\Vert\sigma)}+1-2^{-E_{c}^{k}(\rho
\Vert\sigma)}\\
&  =2^{-\widetilde{E}_{d}^{m}(\rho\Vert\sigma)}+2^{-\widetilde{E}_{c}^{k}%
(\rho\Vert\sigma)}.
\end{align}
This concludes the proof.
\end{proof}

\section{General state pair transformations---Error exponents and strong
converse exponents}

A more general question is to determine the error and strong converse
exponents for general state-pair transformations. Given the state pair
$\left(  \rho,\sigma\right)  $ and the state pair $\left(  \tau,\omega\right)
$, we can define the following error exponent and strong converse exponent:%
\begin{multline}
E^{n,m}(\left(  \rho,\sigma\right)  \rightarrow\left(  \tau,\omega\right)
)\coloneqq\\
-\log_{2}\inf_{\mathcal{P}^{(n)}\in\text{CPTP}}\left\{  \varepsilon
:\mathcal{P}^{(n)}(\rho^{\otimes n})\approx_{\varepsilon}\tau^{\otimes
m},\mathcal{P}^{(n)}(\sigma^{\otimes n})=\omega^{\otimes m}\right\}  ,
\end{multline}%
\begin{multline}
\widetilde{E}^{n,m}(\left(  \rho,\sigma\right)  \rightarrow\left(  \tau
,\omega\right)  )\coloneqq\\
-\log_{2}\left(  1-\inf_{\mathcal{P}^{(n)}\in\text{CPTP}}\left\{
\varepsilon:\mathcal{P}^{(n)}(\rho^{\otimes n})\approx_{\varepsilon}%
\tau^{\otimes m},\mathcal{P}^{(n)}(\sigma^{\otimes n})=\omega^{\otimes
m}\right\}  \right)
\end{multline}
For large $n$, the first one is relevant when $\frac{m}{n}<D(\rho\Vert
\sigma)/D(\tau\Vert\omega)$ and the second one is relevant when $\frac{m}%
{n}>D(\rho\Vert\sigma)/D(\tau\Vert\omega)$. The case of $m=n=1$ was already
considered in \cite[Eq.~(13)]{Wang2019states}, and it was shown therein how
$E^{1,1}(\left(  \rho,\sigma\right)  \rightarrow\left(  \tau,\omega\right)  )$
can be calculated by means of a semi-definite program.

The following bound is a consequence of \cite[Propositions 1 and~2]%
{Wang2019states}:%
\begin{multline}
\frac{1}{n}\widetilde{E}^{n,m}(\left(  \rho^{\otimes n},\sigma^{\otimes
n}\right)  \rightarrow\left(  \tau^{\otimes m},\omega^{\otimes m}\right)  )\\
\geq\max\left\{
\begin{array}
[c]{c}%
\sup_{\alpha\in(0,1)}\left(  \frac{1-\alpha}{2}\right)  \left(  \frac{m}%
{n}\cdot D_{\alpha}(\tau\Vert\omega)-D_{\beta(\alpha)}(\rho\Vert
\sigma)\right)  ,\\
\sup_{\alpha\in(1/2,1)}\left(  \frac{1-\alpha}{2\alpha}\right)  \left(
\frac{m}{n}\cdot\widetilde{D}_{\alpha}(\tau\Vert\omega)-\widetilde{D}%
_{\gamma(\alpha)}(\rho\Vert\sigma)\right)  ,
\end{array}
\right\}  ,
\end{multline}
where%
\begin{align}
&  \beta(\alpha)\coloneqq2-\alpha,\\
&  \gamma(\alpha)\coloneqq\frac{\alpha}{2\alpha-1}.
\end{align}

\section{Recent developments}

\label{sec:rec-devs}

These notes were written in June 2020, and all of the results presented in the
previous sections were developed at that time. Since then, there has been some
interest in the topic of exponents related to smooth max-relative entropy
\cite{LYH21,SD22}, which are clarified here to have operational meaning as
exponents for distinguishability dilution.

In the first paper \cite{LYH21}, the asymptotic error exponent for
distinguishability dilution, when the error is measured using the sine
distance $\sqrt{1-F(\rho,\sigma)} $ \cite{R02,R03,GLN04,R06}, where $F$ is the
fidelity, has been identified (specifically, see \cite[Theorem~6]{LYH21}).
Therein, the asymptotic error exponent for distinguishability dilution is
referred to as the ``exact exponent for the asymptotic decay of the small
modification of the quantum state in smoothing the max-relative entropy.'' It
remains open to identify this quantity when using the normalized trace
distance as the error.

In the second paper \cite{SD22}, a lower bound on the asymptotic strong
converse exponent for distinguishability dilution has been identified (see
\cite[Theorem~2]{SD22}). This bound improves upon the bound given in
Corollary~\ref{cor:sc-exp-dd-cost}. It is easy to state a one-shot version of
the bound given there using the terminology of this note:

\begin{proposition}
[{\cite[Theorem~2]{SD22}}]\label{prop:one-shot-bnd}For states $\rho$ and
$\sigma$ and $m\geq0$, the following inequality holds%
\begin{equation}
\sup_{\alpha\in\left(  0,1\right)  }\left(  \alpha-1\right)  \left(
m-D_{\alpha}(\rho\Vert\sigma)\right)  \leq\widetilde{E}_{c}^{m}(\rho
\Vert\sigma).
\end{equation}

\end{proposition}

\begin{proof}
Consider, by the same approach used for \cite[Theorem~2]{SD22}, that%
\begin{align}
\widetilde{E}_{c}^{m}(\rho\Vert\sigma)  &  =\inf_{\widetilde{\rho}%
\in\mathcal{S},\widetilde{\rho}\leq2^{m}\sigma}\left(  -\log_{2}\left[
1-\frac{1}{2}\left\Vert \widetilde{\rho}-\rho\right\Vert _{1}\right]  \right)
\\
&  =\inf_{\widetilde{\rho}\in\mathcal{S},\widetilde{\rho}\leq2^{m}\sigma
}\left(  -\log_{2}\left[  \inf_{T:0\leq T\leq I}\operatorname{Tr}[\left(
I-T\right)  \rho]+\operatorname{Tr}[T\widetilde{\rho}]\right]  \right) \\
&  \geq-\log_{2}\left[  \inf_{T:0\leq T\leq I}\operatorname{Tr}[\left(
I-T\right)  \rho]+\operatorname{Tr}[T2^{m}\sigma]\right] \\
&  \geq\sup_{\alpha\in\left(  0,1\right)  }\left(  -\log_{2}\operatorname{Tr}%
[\rho^{\alpha}\left(  2^{m}\sigma\right)  ^{1-\alpha}]\right) \\
&  =\sup_{\alpha\in\left(  0,1\right)  }\left(  \alpha-1\right)  \left(
m-D_{\alpha}(\rho\Vert\sigma)\right)  .
\end{align}
This follows by applying \eqref{eq:alt-form-sc-exp-cost}, the constraint
$\widetilde{\rho}\leq2^{m}\sigma$, and \eqref{eq:aud-1}--\eqref{eq:aud-2}.
\end{proof}

Then the asymptotic lower bound from \cite[Theorem~2]{SD22} is a direct
consequence of the one-shot bound from Proposition~\ref{prop:one-shot-bnd}:%
\begin{equation}
\sup_{\alpha\in\left(  0,1\right)  }\left(  \alpha-1\right)  \left(
R-D_{\alpha}(\rho\Vert\sigma)\right)  \leq\lim_{n\rightarrow\infty}\frac{1}%
{n}\widetilde{E}_{c}^{nR}(\rho^{\otimes n}\Vert\sigma^{\otimes n}).
\end{equation}

\bigskip

\textbf{Data availability statement}---Data sharing not applicable to this article as no datasets were generated or analysed during the current study.

\bigskip

\textbf{Acknowledgements}---I acknowledge several discussions with Nilanjana
Datta and Felix Leditzky, and especially their help in establishing
Propositions~\ref{prop:smooth-dmin-to-dist-err-exp} and
\ref{prop:smooth-dmax-to-err-exp-cost}. I also acknowledge Robert Salzmann for discussions.

\bibliographystyle{alpha}
\bibliography{Ref}

\appendix

\section{Background on semi-definite programs}

\label{sec:background-SDPs}A semi-definite program is characterized by a
triple $(\Phi,A,B)$ where $\Phi$ is a Hermiticity-preserving map and $A$ and
$B$ are Hermitian operators. The primal program is given by%
\begin{equation}
\alpha\coloneqq \sup_{X\geq0}\left\{  \operatorname{Tr}[AX]:\Phi(X)\leq
B\right\}  ,
\end{equation}
and the dual program is given by%
\begin{equation}
\beta\coloneqq \sup_{Y\geq0}\left\{  \operatorname{Tr}[BY]:\Phi^{\dag}(Y)\geq
A\right\}  .
\end{equation}

Weak duality is the statement that the following inequality always holds%
\begin{equation}
\alpha\leq\beta.
\end{equation}

Slater's condition for strong duality is as follows:

\begin{enumerate}
\item If there exists $X\geq0$ such that $\Phi(X)\leq B$ and there exists
$Y>0$ such that $\Phi^{\dag}(Y)>A$, then $\alpha=\beta$ and there exists a
primal feasible operator $X$ for which $\operatorname{Tr}[AX]=\alpha$.

\item If there exists $Y\geq0$ such that $\Phi^{\dag}(Y)\geq A$ and there
exists $X>0$ such that $\Phi(X)<B$, then $\alpha=\beta$ and there exists a
dual feasible operator $Y$ for which $\operatorname{Tr}[BY]=\beta$.
\end{enumerate}

Complementary slackness for SDPs is useful for understanding optimal
conditions. Suppose that strong duality holds. Then the following
complementary slackness conditions hold for feasible $X$ and $Y$ if and only
if they are optimal:%
\begin{align}
BY  &  =\Phi(X)Y,\label{eq:comp-slack-1}\\
\Phi^{\dag}(Y)X  &  =AX. \label{eq:comp-slack-2}%
\end{align}

\end{document}